\documentclass{article}
\usepackage{ijcai19}

\usepackage{times}
\usepackage{soul}
\usepackage{url}
\usepackage[hidelinks]{hyperref}
\usepackage[utf8]{inputenc}
\usepackage[small]{caption}
\usepackage{graphicx}
\usepackage{amsmath}
\usepackage{amssymb} 
\usepackage{booktabs}
\usepackage{algorithm}
\usepackage{algorithmic}
\usepackage{subfigure}
\usepackage{xspace}
\usepackage{verbatim}
\usepackage{amsthm}
\usepackage{tikz}
\usepackage{lipsum}

\newtheorem{definition}{\textbf{Definition}}
\newtheorem{theorem}{\textbf{Theorem}}

\newtheorem{example}{\textbf{Example}}

\newcommand{\kc}{$k$-core\xspace}
\newcommand{\kcs}{$k$-cores\xspace}
\newcommand{\kmc}{($k$-1)-core\xspace}
\newcommand{\ks}{$k$-shell\xspace}
\newcommand{\kms}{($k$-1)-shell\xspace}
\newcommand{\akc}{anchored $k$-core\xspace}
\newcommand{\ekc}{edge $k$-core\xspace}
\newcommand{\ekcp}{edge $k$-core problem\xspace}

\newcommand{\clprs}{collapsers\xspace}

\newcommand{\fol}{follower\xspace}
\newcommand{\fols}{followers\xspace}

\newcommand{\cae}{candidate anchor edge\xspace}

\newcommand{\ce}{candidate edge\xspace}
\newcommand{\ces}{candidate edges\xspace}

\newcommand{\ancs}{anchors\xspace}

\newcommand{\ol}{{onion layer}\xspace}
\newcommand{\ols}{{onion layers}\xspace}

\newcommand\blfootnote[1]{%
  \begingroup
  \renewcommand\thefootnote{}\footnote{#1}%
  \addtocounter{footnote}{-1}%
  \endgroup
}

\title{K-Core Maximization through Edge Additions}

\author{{Zhongxin Zhou$^{1,2,4}$\and Fan Zhang$^{1}$\and Xuemin Lin$^{2,3,4}$\and Wenjie Zhang$^{3}$\And Chen Chen$^{2}$}
\\
\affiliations
$^1$Guangzhou University, Guangzhou, China\\
$^2$East China Normal University, Shanghai, China\\
$^3$University of New South Wales, Sydney, Australia\\
$^4$Zhejiang Lab, Hangzhou, China\\
\emails
\{zzxecnu, fanzhang.cs\}@gmail.com, \{lxue, zhangw\}@cse.unsw.edu.au, chenc@zjgsu.edu.cn\\
}

\begin{document}

\maketitle

\begin{abstract}
A popular model to measure the stability of a network is \kc~- the maximal induced subgraph in which every vertex has at least $k$ neighbors.
Many studies maximize the number of vertices in \kc to improve the stability of a network.
In this paper, we study the \ekc problem: Given a graph $G$, an integer $k$ and a budget $b$, add $b$ edges to non-adjacent vertex pairs in $G$ such that the \kc is maximized.
We prove the problem is NP-hard and APX-hard.
A heuristic algorithm is proposed on general graphs with effective optimization techniques.
Comprehensive experiments on 9 real-life datasets demonstrate the effectiveness and the efficiency of our proposed methods.
\end{abstract}

\blfootnote{Zhongxin Zhou and Fan Zhang are the joint first authors. Fan Zhang is the corresponding author.}

\section{Introduction}
\label{sec:intr}

Graphs are widely used to model networks, where each vertex represents a user and each edge represents a connection between two users.
The cohesive subgraph model of \kc, introduced by~\cite{seidman1983network}, is defined as the maximal induced subgraph in which every vertex has at least $k$ neighbors (adjacent vertices) in the subgraph.
The \kc of a network corresponds to the natural equilibrium of a user engagement model: each user incurs a cost (e.g., $k$) to remain engaged but receives a benefit proportional to (e.g., equal to) the number of engaged neighbors.
Since the number of vertices (the size) of \kc reflects the stability of a network,
it is widely adopted in the study of network engagement (stability), e.g.,~\cite{DBLP:journals/siamdm/BhawalkarKLRS15,DBLP:conf/cikm/MalliarosV13,DBLP:conf/wsdm/WuSFLT13}.


To prevent network unraveling, Bhawalkar and Kleinberg et al. propose the \akc problem which maximizes the \kc by anchoring $b$ vertices~\cite{DBLP:journals/siamdm/BhawalkarKLRS15}, where the degree of an anchor is considered as infinitely large.
There are a series of following work to maximize the \kc, e.g.,~\cite{DBLP:conf/icde/ZhangZQZL18,DBLP:journals/pvldb/ZhangZZQL17,DBLP:conf/aaai/ChitnisFG13}.
In order to improve network stability, another basic graph operation is edge addition, which can also be applied to \kc maximization.
Thus, the \ekc problem is proposed~\cite{DBLP:conf/csr/ChitnisT18}: Given a graph $G$,  an integer $k$ and a budget $b$, add $b$ edges to non-adjacent vertex pairs in $G$ such that the \kc is the largest. 

\begin{figure}[t]
	\centering
	\includegraphics[width=0.6\columnwidth, height=24mm]{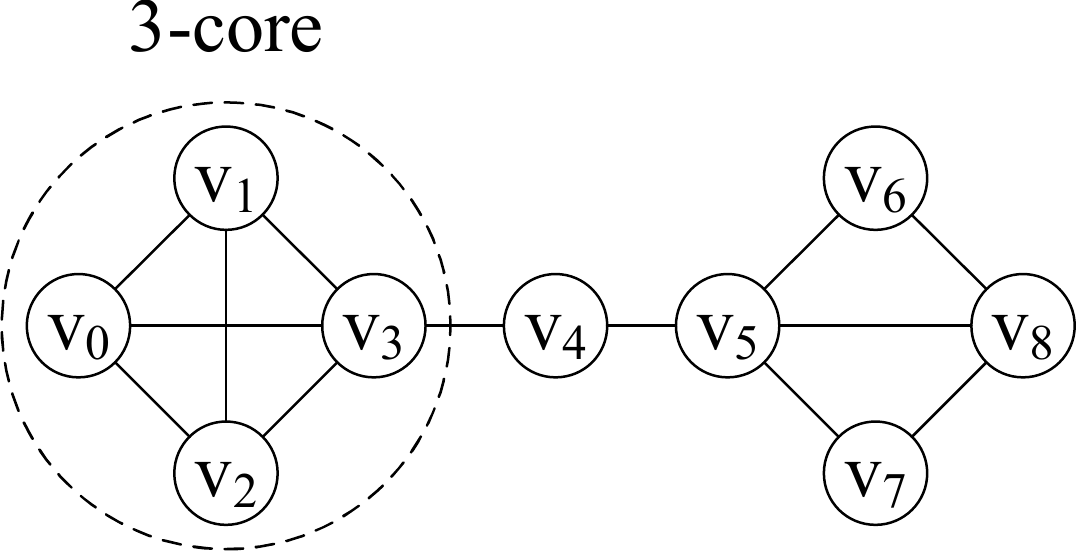}
	\caption{An Example of $k$-Core and Anchoring, $k=3$}
	\label{fig:motivation}
\end{figure}

\begin{example}
Figure~\ref{fig:motivation} depicts a social group $G$ with 9 users and their connections. 
The willingness of a user to keep engaged is influenced by the number of her friends (neighbors) in this group.
According to the \kc model, suppose $k=3$, $v_{4}, v_{6}$, and $v_{7}$ firstly drop out.
Their departure leads to the leave of $v_{5}$ and $v_{8}$, as their degrees decrease to 1 which is less than $k$.
To improve network stability, we can anchor $v_{6}$ and $v_{7}$ based on \akc model, or add an edge between $v_{6}$ and $v_{7}$ based on \ekc model.
Both solutions lead to a larger \kc induced by the vertices in $G$ except $v_4$. 
\end{example}

The \ekc problem can find many applications on real-life networks: friend recommendation in social networks, connection construction in telecom networks, etc.
For instance, in a P2P network, any user benefiting from the network should be connected to at least $k$ other users, to exchange resources. The holder of a P2P network can use the \ekc model to find which connections should be added between users so that a large number of users can successfully use the P2P network~\cite{DBLP:conf/csr/ChitnisT18}.

\paragraph{Challenges and Contributions.} In this paper, we propose a concise reduction to prove that the problem is NP-hard and APX-hard.
The only existing solution is proposed for graphs with bounded tree-width~\cite{DBLP:conf/csr/ChitnisT18}.
However, this assumption usually does not hold in real-life graphs, and their techniques cannot be extended to handle general graphs.
Due to the hardness of the problem, we propose a heuristic algorithm with effective pruning techniques. 
The experiments are conducted on 9 real networks to demonstrate the effectiveness and the efficiency of the proposed methods.

\section{Related Work}
\label{sec:related}

Graph processing on large data may require higher computation efficiency than traditional queries~\cite{DBLP:conf/icde/LuoWL08,DBLP:conf/icde/CheemaBLZW10,DBLP:journals/tkde/LuoWLZWL11}.
Cohesive subgraph mining is a fundamental graph problem, with various models such as clique~\cite{luce1949method}, $k$-core~\cite{seidman1983network}, $k$-fami~\cite{DBLP:conf/dasfaa/ZhangYZQLZ18}, etc.
Among the models, \kc is the only one known to have a linear time algorithm~\cite{DBLP:journals/corr/cs-DS-0310049}.
The \kc has a wide range of applications such as social contagion~\cite{ugander2012structural},
influential spreader identification~\cite{kitsak2010identification},
collapse prediction~\cite{morone2019k},
user engagement study~\cite{DBLP:conf/cikm/MalliarosV13}, etc.

There is an efficient heuristic algorithm for the \akc problem~\cite{{DBLP:journals/pvldb/ZhangZZQL17}}, while it cannot be simply applied to solve the \ekc problem.
One major reason is that the \akc model does not change the topology of the graph while the \ekc model needs to add new edges.
Besides the \kc maximization work introduced in Section~\ref{sec:intr}, there are some studies on \kc minimization under the view of against attack~\cite{DBLP:conf/aaai/ZhangZQZL17,DBLP:conf/cikm/Zhu0WL18,DBLP:journals/corr/abs-1901-02166}.
The edge addition has been studied in different topics.
\cite{DBLP:journals/dam/NatanzonSS01} studies the hardness of edge modification problems on some classes of graphs. \cite{suady2014edge} aims to reduce the diameter of a graph by adding edges. \cite{DBLP:journals/endm/LaiTK05} aims to add a small number of edges in a graph to enlarge the bandwidth.
\cite{DBLP:conf/asunam/KapronSV11} aims to anonymize a given vertex set by adding fewest edges.

\section{Preliminaries}
\label{sec:preli}


We consider a simple, undirected and unweighted graph $G = (V, E)$, where $V$ is a set of vertices and and $E$ is a set of edges.
We denote $n = |V|$, $m = |E|$ and assume $m > n$.
Let $S = (V', E')$ be an induced subgraph of $G$, where $V'\subseteq V$ and $E'\subseteq E$. 
The notations are summarized in Table~\ref{tb:notation}.


\begin{definition}
\label{def:kcore}
\textbf{\kc}.
Given a graph $G$, a subgraph $S$ is the \kc of $G$, denoted by $C_{k}(G)$, if
($i$) $deg(u, S)\geq k$ for every vertice $u \in S$;
($ii$) $S$ is \textit{maximal}, i.e., any subgraph $S'\supset S$ is not a \kc.
\end{definition}

The \kc of a graph $G$ can be obtained by recursively removing every vertex $u$ and its incident edges in $G$ if $deg(u, G) < k$, with a time complexity of $O(m)$.
The \kcs of $G$ with different inputs of $k$ constitute a hierarchical structure of $G$, i.e., $C_{k+1}(G)\subseteq C_{k}(G)$ for every value of $k$.
The definition of \ks is then derived. 

\begin{definition}
\label{def:kshell}
\textbf{\ks}.
Given a graph G, the \ks of G, denoted by $H_{k}(G)$, is the set of vertices in \kc but not in ($k$+1)-core, i.e., $H_{k}(G) = V(C_{k}(G) - C_{k+1}(G))$.
\end{definition}

%

%

If we add some new edges among the vertices which are not adjacent, the \kc of the graph may contain more vertices, which is named the \ekc. The added new edges are called \underline{anchors} or anchor edges. In this paper, we say anchor, add, or insert an edge interchangeably, e.g., an inserted edge is also called an anchored edge. The edges, which may be inserted to the graph, are called \underline{candidate anchors} or candidate edges.

\begin{definition}
\label{def:ekc}
\textbf{\ekc}.
Given a graph $G$ and a set of anchor edges $A\subseteq (\tbinom{V}{2}\setminus E)$, the \ekc, denoted by $C_{k}(G+A)$, is the \kc of the graph $G' = (V, E\cup A)$.
\end{definition}

Due to the addition of anchor edges ($A$), more vertices might be retained in $C_{k}(G+{A})$, in addition to the vertices in $C_k(G)$.
Note that the vertices not incident to the anchor edges may also be retained, according to the contagious nature of \kc computation.
The vertices following the anchor edges $A$ to engage in \kc are named the \underline{followers} of $A$, denoted by $\mathcal{F}(A,G)$.
Formally, $\mathcal{F}(A,G)$ is the set of vertices in $C_{k}(G+A) \setminus C_{k}(G)$. The number of the followers reflects the importance of the corresponding anchor edges. 

%


\paragraph{Problem Statement.} Given a graph $G$, a degree constraint $k$ and a budget $b$, the \ekc problem aims to find a set $A$ of $b$ edges in ${V\choose 2}\setminus E$ such that the number of followers of $A$ is maximized, i.e., $\mathcal{F}(A,G)$ is maximized.
\begin{table}[t]
	\label{sec:preli:notation}
\small
	\centering
	\begin{tabular}{|p{0.24\columnwidth}|p{0.66\columnwidth}|}
		\hline
		\textbf{Notation}	&	\textbf{Definition}	\\ \hline
		$G$	&	an unweighted and undirected graph	\\ \hline
		
        $u$, $v$; $e$, $(u,v)$	&	a vertex in $G$; an edge in $G$  \\ \hline
		
        $m$; $n$ 	&	the number of edges in $G$; the number of vertices in $G$  \\ \hline

		$N(u, G)$	&	the set of adjacent vertices (neighbors) of $u$ in $G$  \\ \hline
		
		$deg(u, G)$	&	the number of adjacent vertices of $u$ in $G$ \\ \hline
		
		$S$	&	a subgraph of $G$	\\ \hline
		
		$V(S)$; $E(S)$	& the vertex set of $S$; the edge set of $S$  \\ \hline
		
        $G[X]$ &  induced subgraph of the vertex set $X$ in $G$  \\ \hline

        $C_{k}(G)$; $H_{k}(G)$	& the \kc of $G$; the \ks of $G$	\\ \hline

        $k$; $b$	&	the degree constraint; the anchor budget \\ \hline

		$A$	&     a set of anchor edges \\ \hline
		
		$G_{A}$; $G_{e}$	& the graph $G+A$; the graph $G+\{e\}$ \\	\hline

        $\mathcal{F}(A, G)$ & the followers of the anchor set $A$ in $G$ \\ \hline
		
		
			

		$\mathcal{L}$; $L_{i}$		&	the \ols of $G$; the $i$-th layer of $\mathcal{L}$	\\ \hline
		

		$l(u)$	&	layer index of $u$ in $\mathcal{L}$	\\	\hline

		$d^{*}(u)$	& number of neighbors of $u$ in its higher layers\\ \hline
		
	\end{tabular}
\vspace{-2mm}
	\caption{Summary of Notations}
	\label{tb:notation}
\end{table}

\section{Complexity}
\label{sec:comp}


\begin{theorem}
\label{the:np-hard}
Edge \kc problem is NP-hard when $k\geq 3$.
\end{theorem}

\begin{proof}
We reduce the \ekc problem from the maximum coverage (MC) problem~\cite{karp1972reducibility} which is NP-hard.
The MC problem is to find at most $b$ sets to cover the largest number of elements, where $b$ is a given budget.
We consider an arbitrary instance of MC with $c$ sets $T_1,..,T_c$ and $d$ elements $\{e_1,..,e_d\} = \cup_{1\leq i\leq c} T_i$.
We suppose $c > k$, $c > b$, and each of the resulting $b$ sets contains at least 2 elements, without loss of generality.
Then we construct a corresponding instance of the \ekc problem on a graph $G$.
Figure~\ref{fig:complexity} shows a construction example from 3 sets and 4 elements when $k=3$.

The set of vertices in $G$ consists of three parts: $M$, $N$ and $Q$.
The part $M$ contains $c$ set of vertices where each set has $d+3$ vertices, i.e., $M = \cup_{1\leq i\leq c} M_i$ where $M_i = \cup_{1\leq j\leq d+3} u^i_j$.
The part $N$ contains $d$ vertices.
The part $Q$ is a ($k$+1)-clique where every two vertices of the $k+1$ vertices are adjacent.
For every $i$ and $j$, if $e_i\in T_j$ in the MC instance, we add an edge between $v_i$ and $u^j_i$. In Figure~\ref{fig:complexity}, these edges are marked in bold.
For every $M_i$ in $M$, we connect $u^i_j$ and $u^i_{j+1}$ by an edge for every $j\in [1,d+2]$, and we also connect $u^i_1$ and $u^i_{d+3}$.
For every $M_i$, we add edges between every vertex in $M_i$ and the vertices in $Q$ so that every vertex in $M_i\setminus\{u^i_{d+1}, u^i_{d+3}\}$ has a degree of $k$ and every vertex in $\{u^i_{d+1}, u^i_{d+3}\}$ has a degree of $k-1$.
We add $k-1$ edges between $v_i$ and the vertices in $Q$ for every $i\in [1,d]$.
The construction of $G$ is completed.

The \kc computation will delete all the vertices in $M$ and $N$. Thus, the \kc of $G$ is $Q$.
To enlarge the \kc, the solution (A) is to add edges between $u^i_{d+1}$ and $u^i_{d+3}$, which is most cost-effective.
Another solution is to add edges between $M_i$ and $M_j$ where $i\neq j$ may also enlarge the \kc, while this solution can always be replaced by solution (A) with at most the same number of \fols.
Adding an edge to a vertex in $N$ is not worthwhile, since we suppose each of the resulting $b$ sets contains at least 2 elements.
Thus, the \ekc problem always chooses $b$ of the $M_i$ in $G$ which corresponds to $b$ sets in the MC problem.
If there is a polynomial time solution for the \ekc problem, the MC problem will be solved in polynomial time.
\end{proof}

\begin{figure}[t]
\small
\begin{center}
\includegraphics[width=1\columnwidth,height=43mm]{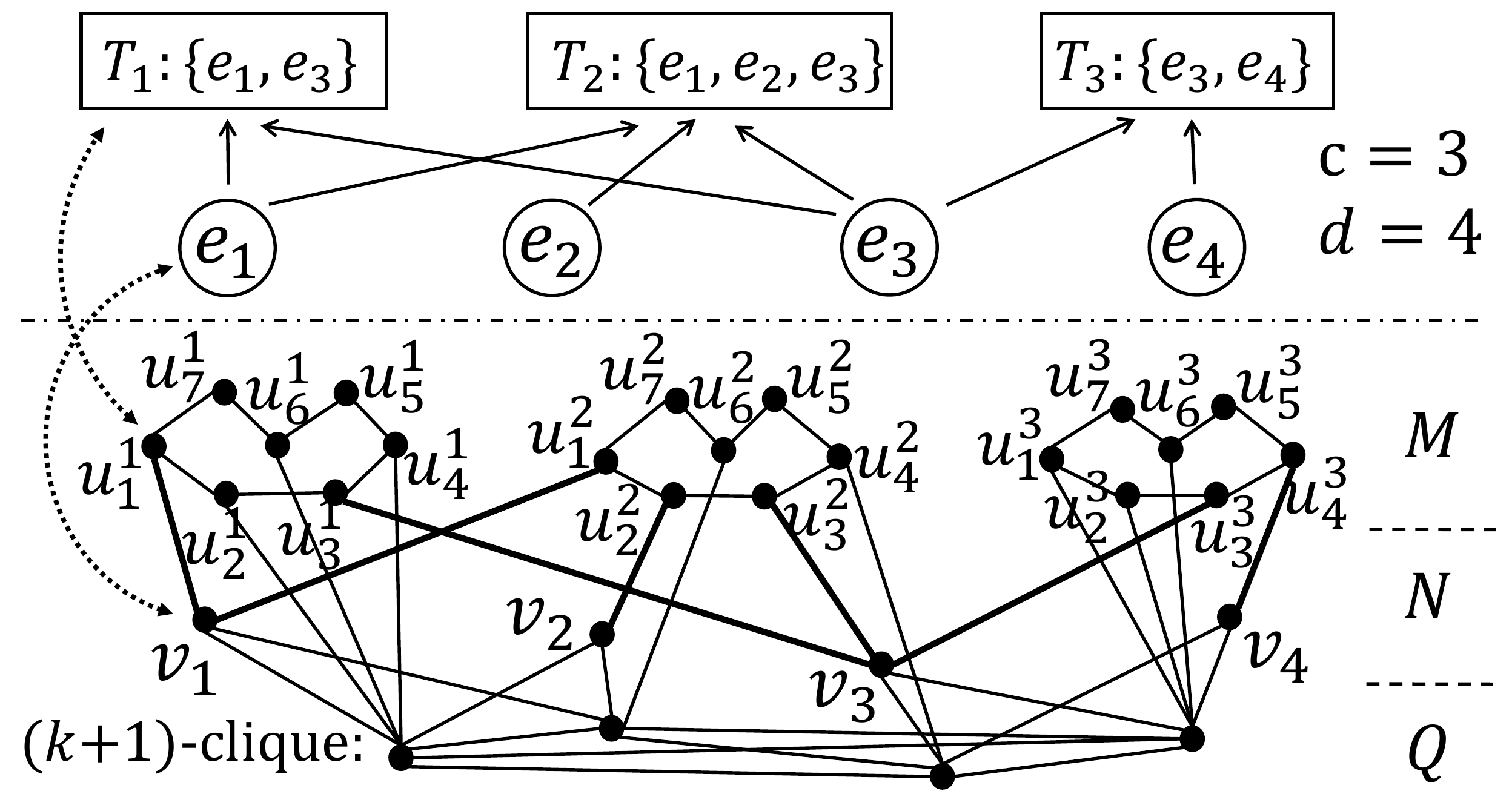}
\end{center}
\vspace{-3mm}
\caption{\small{Complexity Reduction, $k = 3$}}
\label{fig:complexity}
\end{figure}

\begin{theorem}
\label{the:inapprox}
For $k\geq 3$ and any $\epsilon > 0$, the \ekcp cannot be approximated in polynomial time within a ratio of $(1-1/e+\epsilon)$, unless $P=NP$.
\end{theorem}

\begin{proof}
We reduce from the MC problem using a reduction similar to that in the proof of Theorem~\ref{the:np-hard}. For any $\epsilon > 0$, the MC problem cannot be approximated in polynomial time within a ratio of $(1-1/e+\epsilon)$, unless $P=NP$~\cite{DBLP:journals/jacm/Feige98}.
Let $p$ be an arbitrarily large constant.
There are two differences in the construction of $G$: (i) $Q$ is a $p$-clique; and (ii) every $v_i$ is attached by a loop of $p$ vertices where each vertex is connected to $Q$ by $k-2$ edges except $v_i$.
Let $\gamma > 1-1/e$, if there is a solution with $\gamma$-approximation on optimal \fol number for the \ekc problem, there will be a $\lambda$-approximate solution on optimal element number for MC, where $\lambda = \gamma + \frac{(\gamma-1)\times b(d+3)}{p\times f}$ and $f$ is the number of followers of \ekc problem. Thus, the theorem is proved. The \ekc problem is APX-hard.
\end{proof}

\begin{theorem}
\label{the:non-submodular}
Let $f(A) = |\mathcal{F}(A)|$. We have that $f$ is not submodular for $k\geq 2$.
\end{theorem}

\begin{proof}
For two arbitrary \clprs sets $A$ and $B$, if $f$ is submodular, it must hold that $f(A) + f(B) \geq f(A\cup B) + f(A\cap B)$.
Let $Q_1$ be a ($k$+1)-clique where vertices $u$ and $v$ are contained.
Let $Q_2$ be another ($k$+1)-clique.
We create a vertex $w$ and connect it to the vertices in $Q_2$ by $k-2$ edges.
If $A = (u,w)$ and $B = (v,w)$, $f(A) + f(B) = 0 < f(A\cup B) + f(A\cap B) = 1$.
\end{proof}

\section{Solution}
\label{sec:sol}



Due to the NP-hardness and inapproximability of the problem, we resort to a greedy heuristic which iteratively finds the best anchor, i.e., the edge with the largest number of \fols.
The framework of the greedy algorithm is shown in Algorithm~\ref{alg:baseline}.
At Line~\ref{alg:baseline_2}, we compute the \fols for each \cae in the complement graph of $G$, except the anchored edges in $A$ and the edges between the \kc vertices.
Note that adding an edge between two (non-adjacent) \kc vertices cannot enlarge the \kc, because all the non-\kc vertices will still be deleted in \kc computation.
After the computation for every \cae, the best anchor is chosen at Line~\ref{alg:baseline_3}, and the \kc is updated by inserting the anchor edge. The time complexity of Algorithm~\ref{alg:baseline} is $O(b\times n^2\times m)$. As we do not explicitly record the candidate edges, the space complexity of Algorithm~\ref{alg:baseline} is $O(m)$.

\begin{algorithm}[tb]
	\caption{NaiveEKC}
	\label{alg:baseline}
	\textbf{Input}: $G$: a graph, $k$: degree constraint, $b$: budget\\
	\textbf{Output}: a set $A$ of anchor edges
	\begin{algorithmic}[1] 
		\STATE{$A\gets\emptyset$; $C_k\gets C_k(G)$}
		\FOR {$i$ from $1$ to $b$}
		{
			\FOR {each $e \in {V\choose 2}\setminus \{E\cup A\cup {V(C_k)\choose 2}\}$}
			\label{alg:baseline_1}
			{
				\STATE compute $\mathcal{F}(e, G+A)$;
				\label{alg:baseline_2}
			}
			\ENDFOR
			\STATE $e^{*} \gets$ the edge with most \fols; \label{alg:baseline_3}
			\STATE $A \gets A\cup e^{*}$; $C_k\gets C_k(G+A)$; \label{alg:baseline_4}
		}
		\ENDFOR
		\RETURN{$A$}
	\end{algorithmic}
\end{algorithm}






\subsection{Optimizations on Each Iteration}
\label{sec:sol:reduceces}


In this section, we introduce optimization techniques for the first iteration of the greedy algorithm, i.e., the \ekc problem with $b = 1$.
They can be immediately applied to other iterations by replace the \kc of $G$ by the \ekc of $G+A$.


\vspace{2mm}
\noindent
\textbf{Basic Candidate Pruning}
\vspace{1mm}

\noindent
The following theorem locates the scope of valid candidate anchors where each edge has at least one follower.



\begin{theorem}
\label{theorem:baseline_ce}
Given a graph $G$, if a \ce $e=(u, v)$ has at least one follower, we have that $u\in C_{k-1}(G)$ and $v\in C_{k-1}(G)$, where at least one of $u$ and $v$ is in $H_{k-1}(G)$.
\end{theorem}

\begin{proof}
Suppose $u\notin C_{k-1}(G)$. We have $e\in C_{k}(G_e)$; otherwise, we have $C_{k}(G_e) = C_k(G)$, i.e., there is no follower of $e$.
If we remove $e$ from $C_{k}(G_e)$, then $deg(u, C_{k}(G_e)\setminus\{e\})$ and $deg(v, C_{k}(G_e)\setminus\{e\})$ are at least $k-1$, because $(u,v)\in C_k(G_e)$ and their degrees decrease by 1.
Thus, $C_{k}(G_e)\setminus\{e\}\subseteq C_{k-1}(G)$ which contradicts with $u\notin C_{k-1}(G)$.
Now we have proved that $u\in C_{k-1}(G)$ and $v\in C_{k-1}(G)$.
Suppose that $\{u, v\}\in C_k(G)$ and $C_k(G_e)\setminus C_k(G) = F \neq\emptyset$, we have that $C_k(G_e)\setminus\{e\}$ belongs to $C_k(G)$, which contradicts with $F\notin C_k(G)$.
\end{proof}

\begin{example}
In Figure~\ref{fig:example}, when $k=3$, the 3-core $C_{3}(G)$ is induced by $\{v_{0}, v_{1}, v_{5}, v_{6}\}$, and the 2-shell $H_{2}(G)=\{v_{2}, v_{3}, v_{7}, v_{8}\}$.
According to Theorem~\ref{theorem:baseline_ce}, we only need to consider every new edge $(u,v)$ or $(v,u)$ with $u\in H_2(G)$ and $v\in V(C_2(G)) = V(C_3(G))\cup H_2(G)$, as candidates.
\end{example}


\vspace{2mm}
\noindent
\textbf{Onion Layer based Candidate Pruning}
\vspace{1mm}

\noindent
Given the \kmc of $G$, the computation of \kc on $G$ recursively deletes some vertices in the \kmc. The vertices are deleted in batch as their degrees are less than $k$ at a same time, like peeling an onion~\cite{DBLP:conf/icde/ZhangYZQ17,DBLP:journals/pvldb/ZhangZZQL17}. The first layer of the onion, denoted by $L_1$, consists of the vertices in \kmc with degree less than $k$, i.e., $L_1 = \{u~|~deg(u, C_{k-1}(G)) < k\}$. The deletion of the first layer vertices may decrease the degrees of other vertices, and produces the second layer.
Recursively, we have that $L_i = \{u~|~deg(u, G_i) < k\}$ where $G_i = C_{k-1}(G) - G[\cup_{1\leq j < i} L_j]$.


Let $l(u)$ denote the layer index of $u$, i.e., $u\in L_{l(u)}$. If $l(u) < l(v)$, we say $u$ is at a lower layer of $v$, or $v$ is at a higher layer of $u$.
We use $\mathcal{L}$ to denote the union of layers, i.e., $\mathcal{L} = \cup_{1\leq i\leq s} L_i$ where $s = max(\{l(u)~|~u\in H_{k-1}(G)\})$.
Let $d^*(u)$ denote the number of adjacent neighbors of $u$ in higher layers and \kc, i.e., $d^*(u) = deg(u, G')$ where $G' = C_{k-1}(G) - G[\{\cup_{1\leq i\leq l(u)} L_i\}\setminus u]$.

Benefit from $\mathcal{L}$, we propose an effective pruning technique. 


\begin{theorem}
\label{theorem:reduce_olprun}
Given a graph $G$, if a \ce $e=(u, v)$ has at least one follower, we have that
($i$) $d^{*}(u)= k-1$ when $l(u)<l(v)$,
($ii$) $d^{*}(v)= k-1$ when $l(v)<l(u)$, and
($iii$) $d^{*}(u)=d^{*}(v)=k-1$ when $l(u)=l(v)$;
\end{theorem}

\begin{proof}

Let $O$ be the vertex deletion order of computing the \kc on the \kmc.
We have $d^{*}(u) \leq k-1$ and $d^{*}(v)\leq k-1$, since $u$ and $v$ are deleted at their layers in $O$, respectively.
When $l(u)<l(v)$, we suppose $d^{*}(u) < k-1$.
The anchoring of $(u,v)$ increases the degrees of $u$ and $v$ by 1, if we use the same order $O$ to compute the \kc again, then $d^{*}(u)\leq k-1 < k$, i.e., every vertex will be deleted at the same position of the order $O$, including $u$ and $v$. It contradicts with $e = (u,v)$ has at least one follower.
Thus, case ($i$) and ($ii$) are proved.
When $l(u)=l(v)$, we suppose $d^{*}(u) < k-1$ without loss of generality. If we use the same order $O$ to compute the \kc again with the anchored edge $e$, $u$ will be deleted at the same position since $d^{*}(u) < k$. So $v$ can't survive after the next batch of deletions since $d^{*}(v) < k$. Finally, every vertex will be deleted.
\end{proof}


\begin{figure}[t]
	\centering
	\includegraphics[width=0.57\columnwidth,height=25mm]{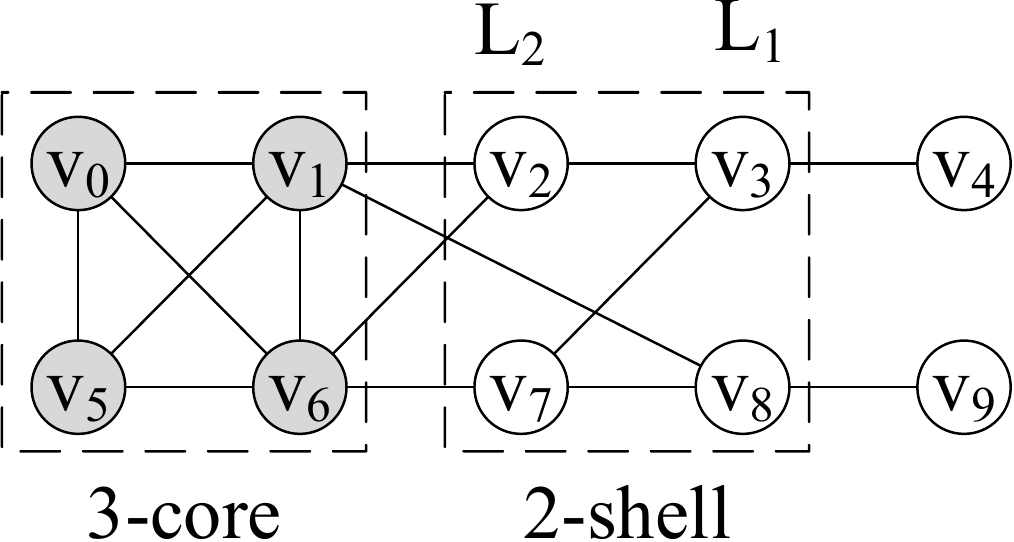}
	\vspace{-1mm}
	\caption{Candidate Seletion, $k=3$}
	\label{fig:example}
\end{figure}

\begin{example}
In Figure~\ref{fig:example}, when $k=3$, layer 1 contains $v_{3}$ and $v_{8}$, and layer 2 contains $v_{2}$ and $v_{7}$.
By theorem~\ref{theorem:reduce_olprun}, ($v_{2}, v_{7}$) is not a promising \ce because $d^{*}(v_{7})=1$, but ($v_{3}, v_{8}$) is a proper \ce.
\end{example}

Note that if $l(u)\leq l(v)$, in the \kc computation with anchor $e$, the survive of $u$ may preserve some vertices before visiting $v$. Thus, in Theorem~\ref{theorem:reduce_olprun}, there is no degree requirement for $v$ to ensure that $e = (u,v)$ has at least one follower.


\vspace{2mm}
\noindent
\textbf{Onion Layer based Follower Computation}
\vspace{1mm}

\noindent
A naive follower computation is to directly apply the \kc computation on the graph with the existence of an anchor.
An improved idea is to use the core maintenance algorithms, which update the \kc for \textit{every} $k$ with the addition of an edge~\cite{DBLP:conf/icde/ZhangYZQ17}.
For the \ekc algorithm, we only need to update the \kc of a given $k$ with an anchor edge.
Thus, we adopt and refine the follower computation in the vertex-anchored \kc algorithm (OLAK) which is shown to be more efficient than core maintenance for \kc update with a fixed $k$~\cite{DBLP:journals/pvldb/ZhangZZQL17}.

In OLAK, anchoring a vertex means the degree of the vertex is infinitely large.
There is an observation that a vertex $u$ is the follower of the anchor vertex $x$ if there is a path $x\leadsto u$ based on neighboring relations, and $l(y)<l(z)$ for every two consecutive vertices $y$ and $z$ along the path. This indicates that we do not need to consider the vertices without such paths in the follower computation.
Different from \akc problem, \ekc problem requires us to add a new edge, where the vertex degree is increased by exactly 1, and we have to consider two vertices rather than one.

In our algorithm, given an anchor edge $(u,v)$ with $l(u) < l(v)$, we generate candidate \fols layer-by-layer through activating the neighbors at higher layers, starting from the layer of $u$. Here $u$ is the first activated vertex and only the activated vertices can activate their neighbors at higher layers.
In this activation procedure, each activated vertex is assigned an upper bound of its degree in the \ekc. The degree upper bound of a vertex $u$ is generated by counting its neighbors in higher layers and the activated neighbors of $u$ at other layers. Once the degree upper bound of a vertex is less than $k$, it will be deleted immediately and the upper bounds of its neighbours will be decreased by 1.
The follower computation is complete when the layer-by-layer activation is finished. Note that the two vertices incident to the anchor edge may also be deleted in this procedure. Once one of them is deleted, there is no follower and the computation is returned.

\vspace{2mm}
\noindent
\textbf{Follower based Candidate Pruning}
\vspace{1mm}


\noindent
We can further prune some candidate edges by known results, when we retrieve the followers of some candidate edges.

\begin{theorem}
\label{theorem:reduce_folprun}
Given a \cae $e_{1} = (u_{1},v_{1})$ and its follower set $\mathcal{F}(e_{1})$, we have $\mathcal{F}(e)\subseteq \mathcal{F}(e_{1})$ for every candidate edge $e\in \{(u,v)~|~u\in \mathcal{F}(e_{1})~$and$~v\in \mathcal{F}(e_{1})\cup C_k(G)~$and$~u\neq v\}$.
\end{theorem}

\begin{proof}
Let $e=(u,v)$. If $\mathcal{F}(e)$ is not empty, then $C_k(G_e)$ contains $C_k(G)$, $u$, $v$ and $\mathcal{F}(e)$. Because $C_k(G_{e_1})$ contains $u$ and $v$, i.e., more vertices may be added to $C_k(G_e)$ after anchoring $e_1$. So, $V(C_k(G_e))\subseteq V(C_k(G_{e_1}))$ and thus $\mathcal{F}(e)\subseteq \mathcal{F}(e_1)$.
\end{proof}

\begin{example}
In Figure~\ref{fig:example}, when $k=3$,
if we get that $v_{2}$ and $v_{7}$ are the \fols of edge ($v_{3}, v_{8}$). According to Theorem~\ref{theorem:reduce_folprun}, the \fols of $(v_{2}, v_{7})$ are also the \fols of ($v_{3}, v_{8}$), thus $(v_{2}, v_{7})$ is not a promising candidate in current iteration.
\end{example}


\begin{algorithm}[t]
	\caption{EKC}
	\label{alg:ekc}
	\textbf{Input}: $G$: a graph, $k$: degree constraint, $b$: budget\\
	\textbf{Output}: A set $A$ of anchor edges
	\begin{algorithmic}[1] 
		\STATE $A\gets\emptyset$;
		\FOR {$i$ from $1$ to $b$}
		{
            \STATE $N\gets \{(u,v)~|~u\in H_{k-1}(G+A), v\in C_{k-1}(G+A)\}\setminus E(C_{k-1}(G+A))$ (Theorem~\ref{theorem:baseline_ce});
            	\label{alg:ekc_1}
			\STATE compute $\mathcal{L}$ and filter $N$ based on Theorem~\ref{theorem:reduce_olprun};
			\label{alg:ekc_2}
			\FOR {each $e \in N$}
			\label{alg:ekc_x}
			{
				\STATE $\mathcal{F}(e)\gets$ FindFollower$(e, \mathcal{L})$ (Section~\ref{sec:sol:reduceces});
				\STATE update $N$ based on Theorem~\ref{theorem:reduce_folprun};
				\label{alg:ekc_3}
			}
		    \ENDFOR
			\STATE $e^{*} \gets$ the edge with most \fols; $A \gets A\cup e^{*}$;
			\STATE update $N$ (Section~\ref{sec:sol:reuse});	
		}
		\ENDFOR
		\RETURN{$A$}
	\end{algorithmic}
\end{algorithm}
\subsection{Reusing Intermediate Results across Iterations}
\label{sec:sol:reuse}

When one iteration of the greedy algorithm is completed, we get the best anchor edge $A$ and the number of followers for every candidate edge in this iteration.
These results can be reused since some connected components in the induced subgraph of \kms may keep the same topology after anchoring an edge.
For each connected component $S$ where none of the vertices are incident to the anchor $A$, we can record the largest number of followers of one candidate anchor in $S$, for the later iterations.
For the connected component(s) $S$ where there is a vertex incident to the anchor $A$, it is hard to reuse the results because the addition of edges may largely change the vertex deletion order and the layer structure in \kc computation.

\subsection{EKC Algorithm}
\label{sec:sol:ekcalgo}

We present the EKC algorithm in Algorithm~\ref{alg:ekc}, which optimizes the greedy algorithm by adopting all the proposed techniques.
At Line~\ref{alg:ekc_1}, the \ce set is restricted to $N$ according to Theorem~\ref{theorem:baseline_ce}.
In Line~\ref{alg:ekc_2}, we compute the \ols of $G$ and then exclude the candidates  
based on Theorem~\ref{theorem:reduce_olprun}.
The follower computation of a candidate edge is conducted by exploring $\mathcal{L}$ layer-by-layer, as introduced in Section~\ref{sec:sol:reduceces}.
The set $N$ is further filtered by Theorem~\ref{theorem:reduce_folprun} once a follower computation is completed.
After each iteration, we get the anchor edge with the most followers.
The algorithm terminates after $b$ iterations.




The resulting set of $b$ anchors is same to that in Algorithm~\ref{alg:ekc}. The correctness is guaranteed by the correctness of optimization techniques. The worst-case time complexity and space complexity are same to Algorithm~\ref{alg:baseline}. 

\section{Evaluation}
\label{sec:exp}


\paragraph{Algorithms.} As far as we know, there is no existing algorithm for the \ekc problem on general graphs. In the experiment, we implement and evaluate 11 algorithms as shown in Table~\ref{tb:algorithms}, where the bottom 6 algorithms produce the same result, because they use the same greedy heuristic and candidate visiting order.

\begin{table}[t]
\small
	\centering
	\begin{tabular}{|p{0.16\columnwidth}|p{0.74\columnwidth}|}
		\hline
		\textbf{Algorithm}   & \textbf{Description}\\ \hline
		\texttt{Rand}  &  randomly chooses $b$ anchor edges (each with at least one follower) from $N_1 = \{(u,v)~|~u\in H_{k-1}(G), v\in C_{k-1}(G)\}\setminus E(C_{k-1}(G))$\\ \hline
		\texttt{Degree} & chooses $b$ \ancs from $N_{1}$ with largest degrees in $\mathcal{L}$\\ \hline
		\texttt{Layer} & chooses $b$ \ancs from $N_{1}$ at highest onion layers in $\mathcal{L}$\\ \hline
        \texttt{AKC}  &  the anchored \kc algorithm in~\cite{DBLP:journals/pvldb/ZhangZZQL17} \\ \hline
		\texttt{Exact}  &  identifies the optimal solution by exhaustively searching all possible combinations of $b$ anchors, with all the proposed techniques\\ \hline
		\texttt{Naive}  &  computes a \kc on $G$ for each candidate to find best anchor in each iteration (Algorithm~\ref{alg:baseline})\\ \hline
		\texttt{Baseline}  &  \texttt{Naive} + filtering candidates with Theorem~\ref{theorem:baseline_ce}\\ \hline
		\texttt{BL+O1}  &  \texttt{Baseline} + filtering candidates with Theorem~\ref{theorem:reduce_olprun} \\ \hline
		\texttt{BL+O2}  &  \texttt{BL+O1} + filtering candidates with Theorem~\ref{theorem:reduce_folprun}\\ \hline
		\texttt{BL+OF}  &  \texttt{BL+O2} + $\mathcal{L}$ based \fol computation (Section~\ref{sec:sol:reduceces})\\ \hline
		\texttt{EKC}  &  \texttt{BL+OF} + intermediate result reuse (Section~\ref{sec:sol:reuse})\\ \hline
	\end{tabular}
	\vspace{-2mm}
	\caption{Summary of Algorithms}
	\label{tb:algorithms}
\end{table}

%
%



\begin{table}[t]
\small
  \centering
    \begin{tabular}{|l|l|l|l|l|}
      \hline
      \textbf{Dataset}   & \textbf{Vertices}  & \textbf{Edges}  & $d_{avg}$ & $k_{max}$\\ \hline
      \texttt{Facebook}  &  4,039 & 88,234 & 43.69 & 115\\ \hline
      \texttt{Enron} & 36,692 & 183,831 & 10.02 & 43\\ \hline
	  \texttt{Brightkite}  &  58,228 & 214,078 & 7.35 & 52\\ \hline
      \texttt{Gowalla}  & 196,591 & 950,327 & 9.67 & 51\\ \hline
      \texttt{DBLP} &  317,080 & 1,049,866 & 6.62 & 113\\ \hline
	  \texttt{Twitter} & 81,306 & 1,768,149 & 33.02 & 96\\ \hline
      \texttt{Stanford}	& 281,903 & 2,312,497 & 14.14 & 71\\ \hline
      \texttt{YouTube}  & 1,134,890 & 2,987,624 & 5.27 & 51\\ \hline
      \texttt{Flickr}  & 513,969 & 3,190,452 & 12.41 & 309\\ \hline
    \end{tabular}
\vspace{-2mm}
\caption{Statistics of Datasets}
\label{tb:datasets}
\end{table}

\paragraph{Datasets.}
\texttt{Flickr} is from \url{http://networkrepository.com/}, and the others are from \url{https://snap.stanford.edu/data/}.
Table~\ref{tb:datasets} shows the statistics of the datasets. 
%

\paragraph{Settings.} All programs are implemented in C++.
All experiments are performed on Intel Xeon 2.20GHz CPU and Linux System.
We vary the parameters $k$ and $b$.


\subsection{Effectiveness}
\label{sec:exp:effectiveness}




\begin{figure}[t]
	\centering \includegraphics[width=0.8\columnwidth,height=20mm]{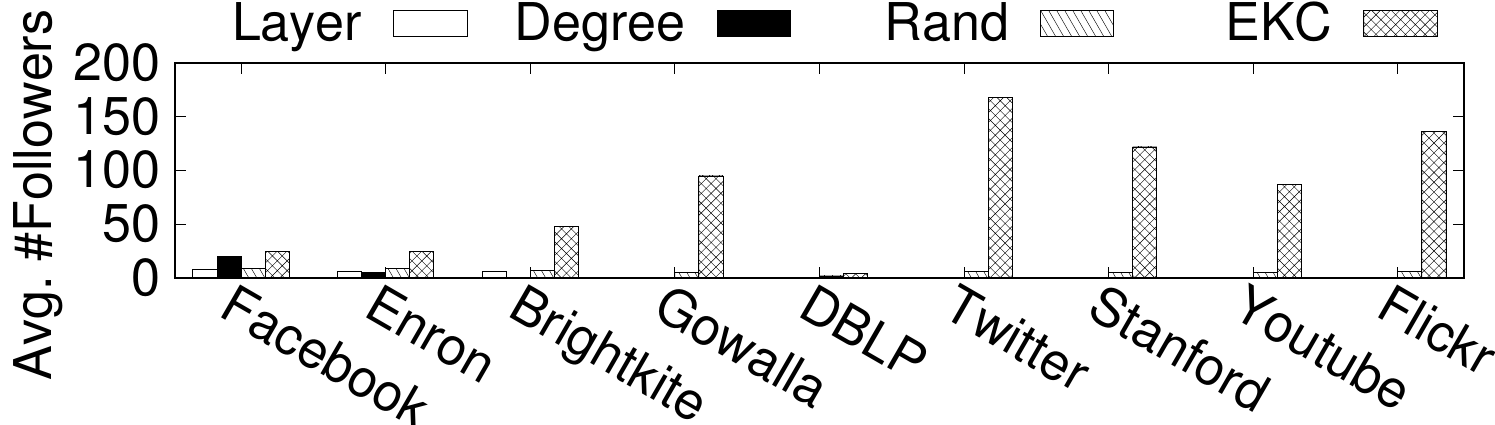}
	\vspace{-2mm}
	\caption{Effectiveness of the Greedy Heuristic, k=20, b=5}
	\label{fig:effect_rand}
\end{figure}

Figure~\ref{fig:effect_rand} reports the number of \fols w.r.t. $b$ \ancs. 
For random approaches, we report the average number of \fols for 100 independent tests. 
In Figure~\ref{fig:effect_rand}, \texttt{Degree} and \texttt{Layer} failed to get any \fols in more than half of the datasets, because a large degree vertex or a vertex in higher \ol does not necessarily have a \fol.
A totally random algorithm cannot get any \fols in most settings.
Thus, in \texttt{Rand}, we only choose the edges from the set where each edge has at least one \fol.
We observe that the majority of the \ces does not have any \fols, while our greedy heuristic always finds effective \ancs in the experiments. 
In Figure~\ref{fig:effect_rand}, we notice that \texttt{EKC} preserves more than 100 \fols in \texttt{Flickr} with only $5$ \ancs.


\begin{figure}[t]
	\centering
	\subfigure[Gowalla, b=3] { \includegraphics[width=0.47\columnwidth,height=22mm]{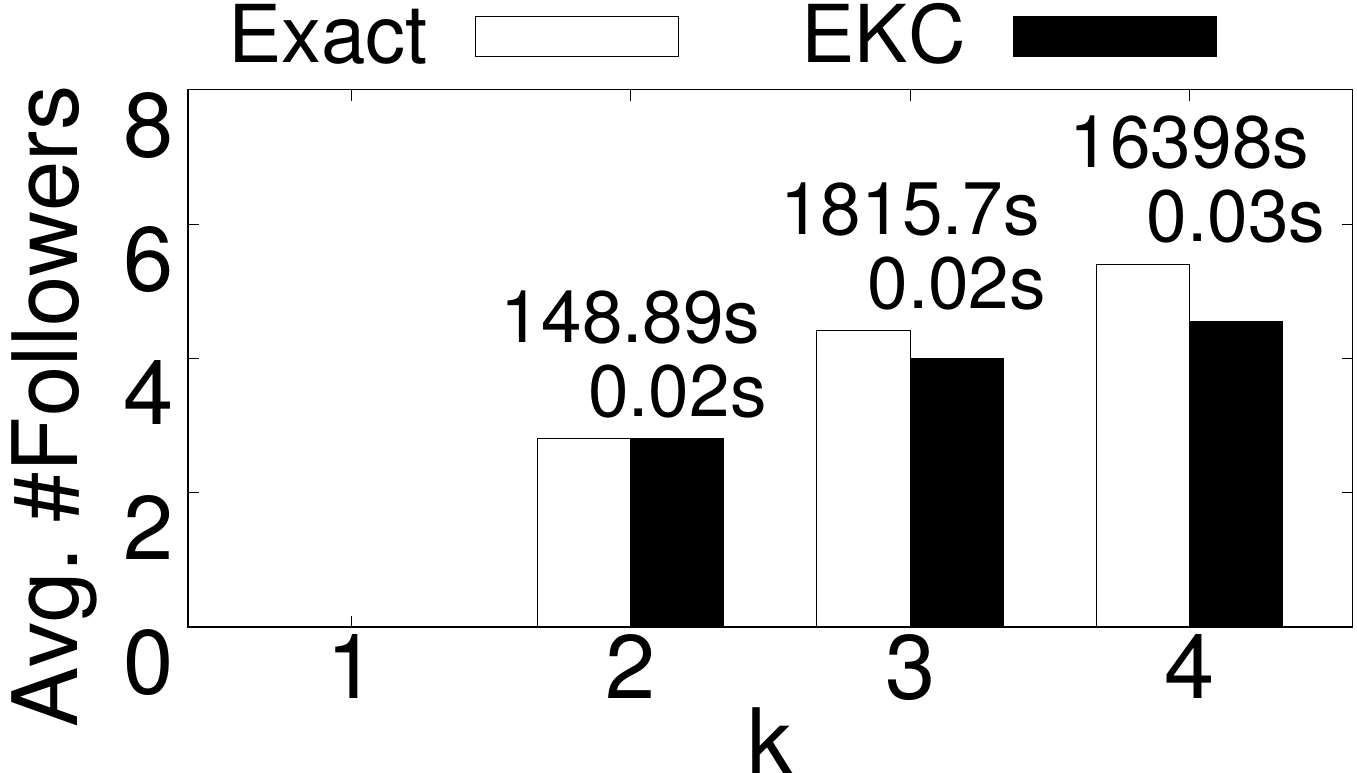}
		\label{fig:exactvsekc_1}
	}
	\subfigure[Brightkite, k=3] {
		\includegraphics[width=0.47\columnwidth,height=22mm]{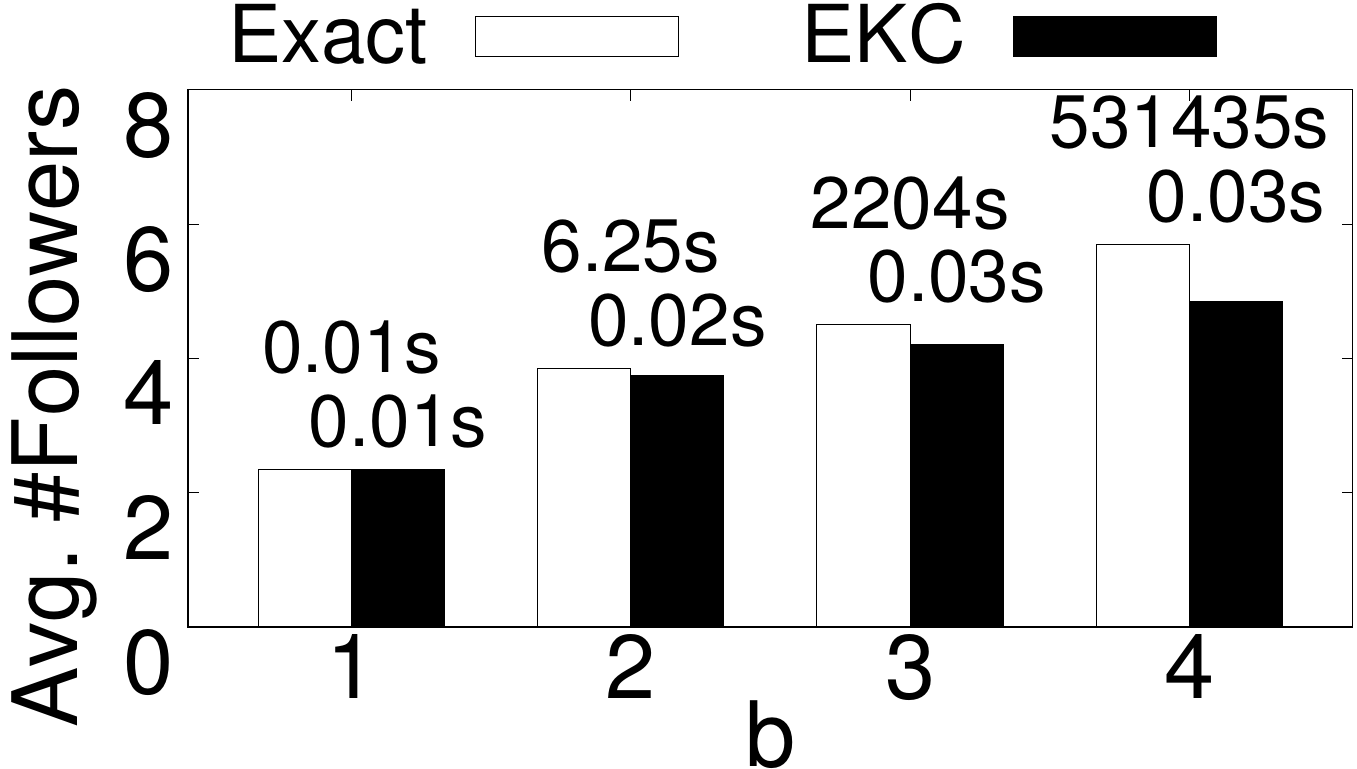}
		\label{fig:exactvsekc_2}
	}
\vspace{-3mm}
	\caption{Exact vs EKC}
	\label{fig:exactvsekc}
\end{figure}

We also compare the performance of \texttt{EKC} with the optimal solution from \texttt{Exact}. 
Due to extremely high cost, we run \texttt{Exact} on an induced subgraph by 50 random vertices.
The results are from 100 independent settings.
The values of $k$ and $b$ are small due to the small-scale of data.
Figure~\ref{fig:exactvsekc} shows that the margins between \texttt{EKC} and \texttt{Exact} are not unacceptable, considering the non-submodular property and the significantly better efficiency of \texttt{EKC}.



\begin{figure}[b]
	\centering
	\subfigure[Gowalla, k=20] {
		\includegraphics[width=0.47\columnwidth,height=22mm]{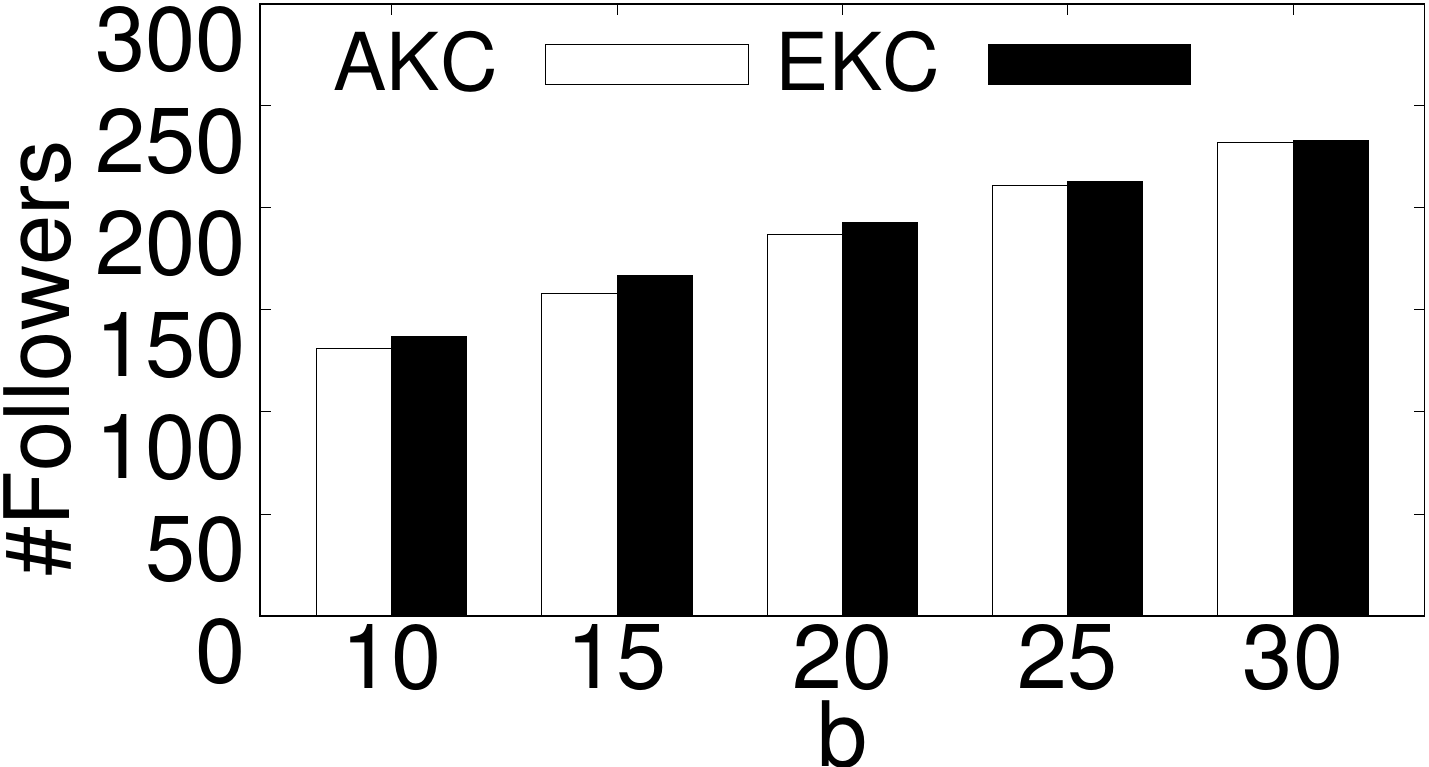}
	}
	\subfigure[Brightkite, b=5] {
		\includegraphics[width=0.47\columnwidth,height=22mm]{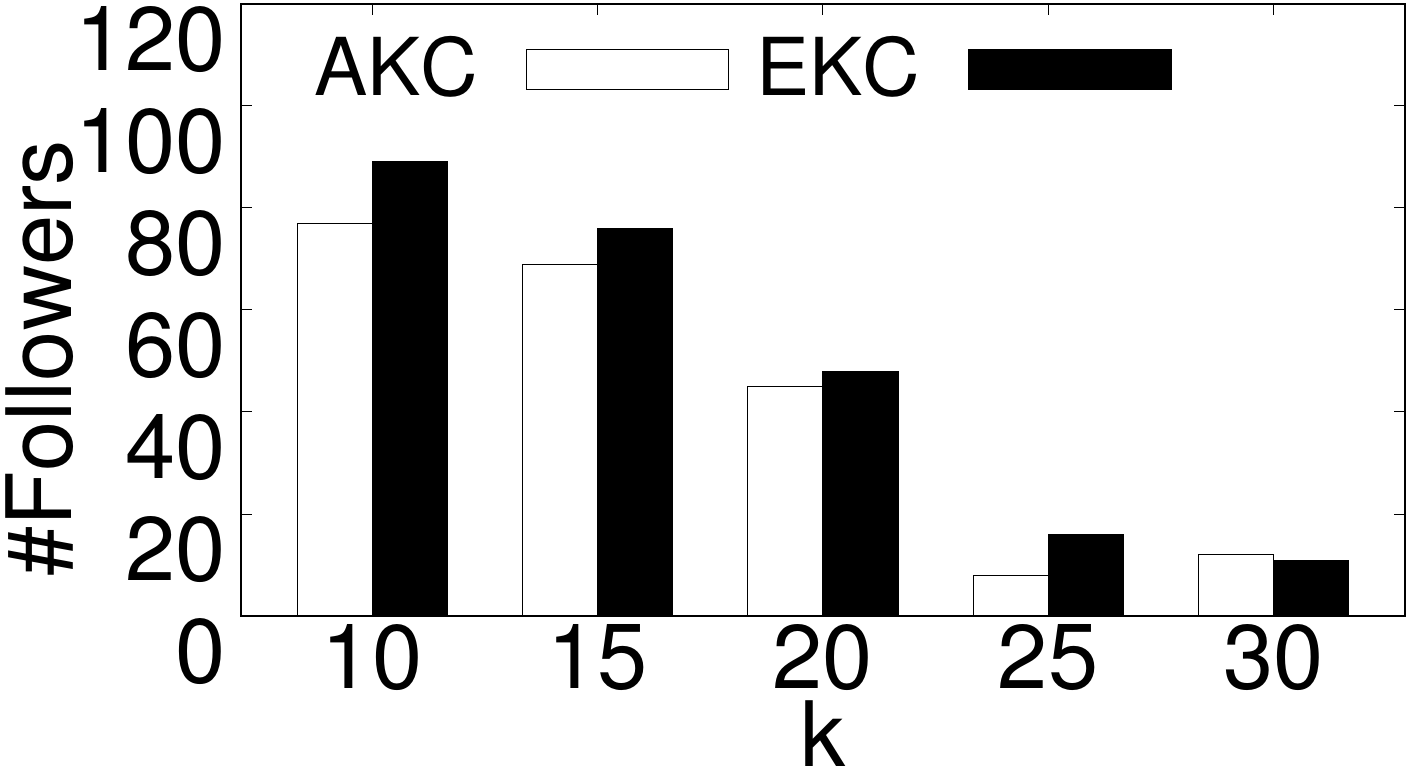}
	}
\vspace{-3mm}
	\caption{Anchor \kc vs Edge \kc}
	\label{fig:ekcvsanchor}
\end{figure}

Figure~\ref{fig:ekcvsanchor} shows that \texttt{EKC} produces similar and almost larger numbers of \fols than \texttt{AKC} where the meaning of $b$ is different: the former for edges and the latter for vertices.
Although the anchoring of one vertex means the preserving of many edges, \texttt{EKC} can produce similar \fols with one anchor edge.
It shows that the \ekc can enlarge the \kc with less graph manipulations. 


\subsection{Efficiency}
\label{sec:exp:efficiency}

\begin{figure}[t]
	\centering \includegraphics[width=0.8\columnwidth,height=20mm]{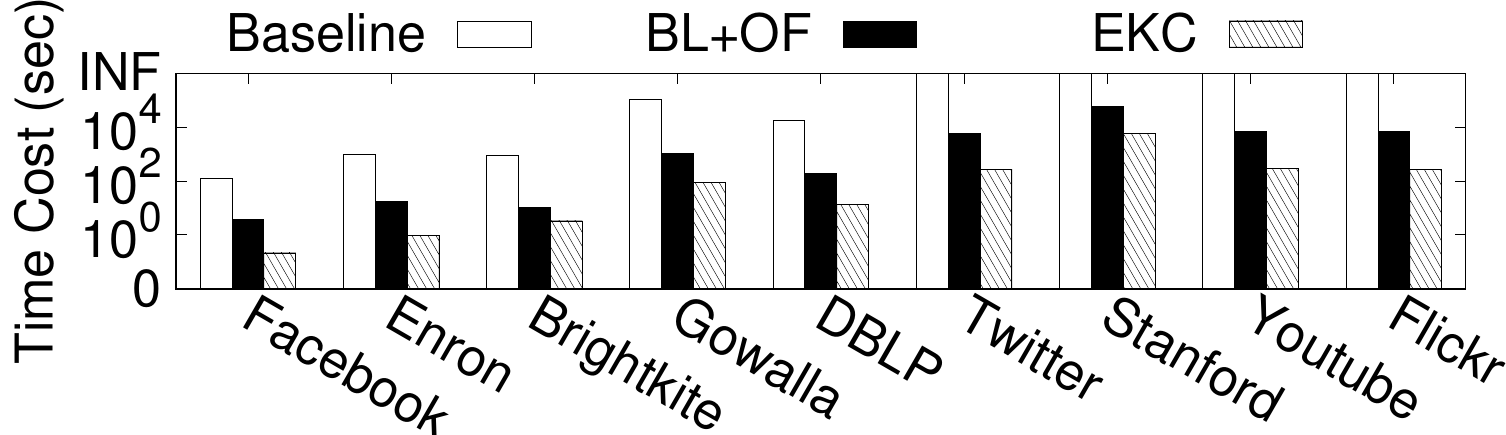}
	\vspace{-2mm}
	\caption{Running time on Different Datasets}
	\label{fig:multi-dataset}
\end{figure}
Figure~\ref{fig:multi-dataset} reports the performance of three algorithms on all the datasets with $k = 20$ and $b = 5$.
The datasets are ordered by the number of edges.
We find the runtime among different datasets is largely influenced by the characteristics of datasets and the onion layer structures in \kc computation. 
\texttt{Baseline} cannot finish computation on 4 datasets within one week.
\texttt{BL+OF} significantly improve the performance by applying a series of techniques. \texttt{EKC} performs even better given all the optimizations.


\begin{figure}[t]
	\centering
	\includegraphics[width=0.9\columnwidth]{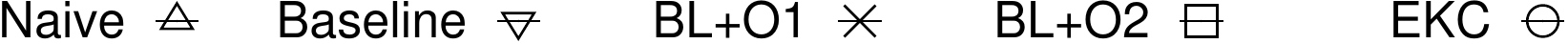}\\
	\begin{minipage}{0.45\columnwidth}
		\subfigure[Gowalla, b=20] {%
			\includegraphics[width=\columnwidth,height=20mm]{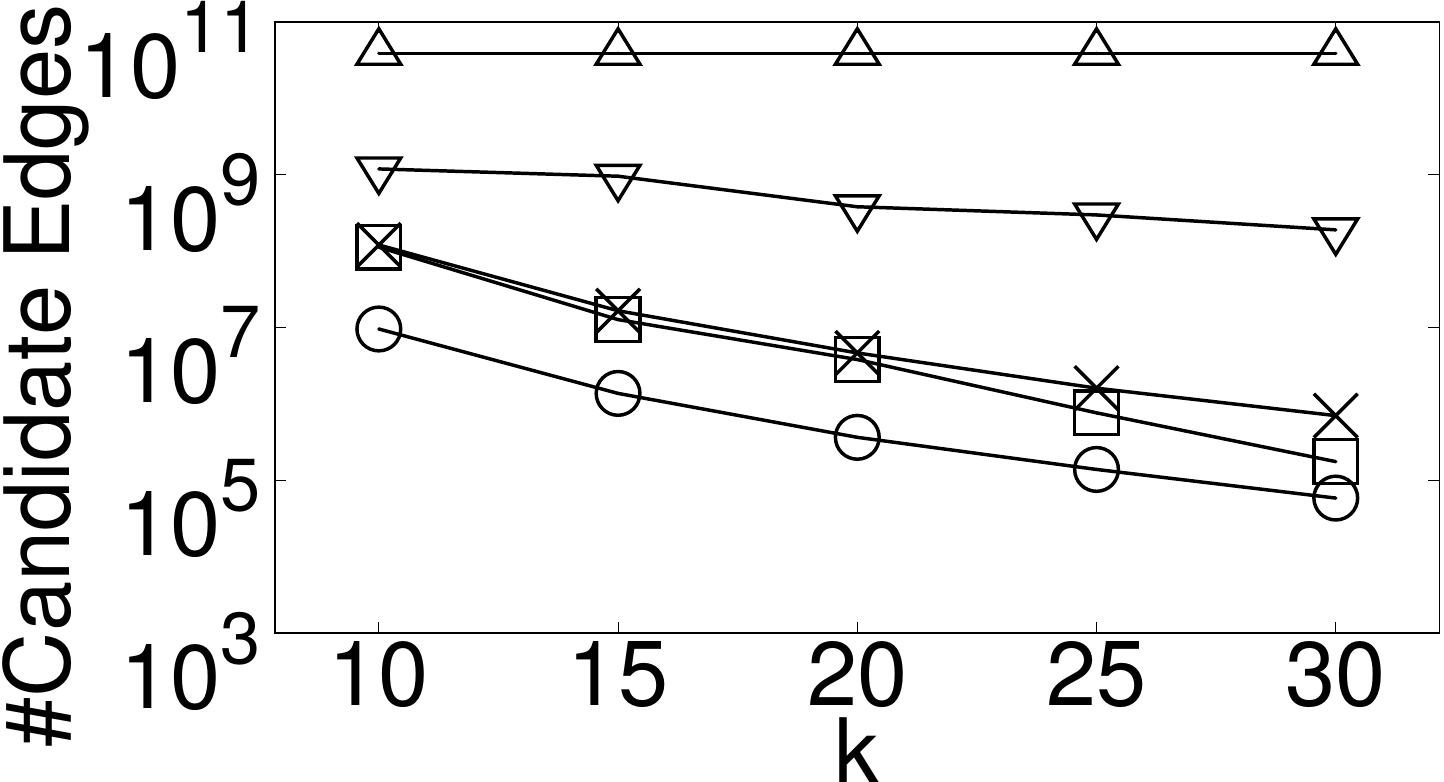}
		}
	\end{minipage}
	\begin{minipage}{0.45\columnwidth}
		\subfigure[Brightkite, b=5] {
			\includegraphics[width=\columnwidth,height=20mm]{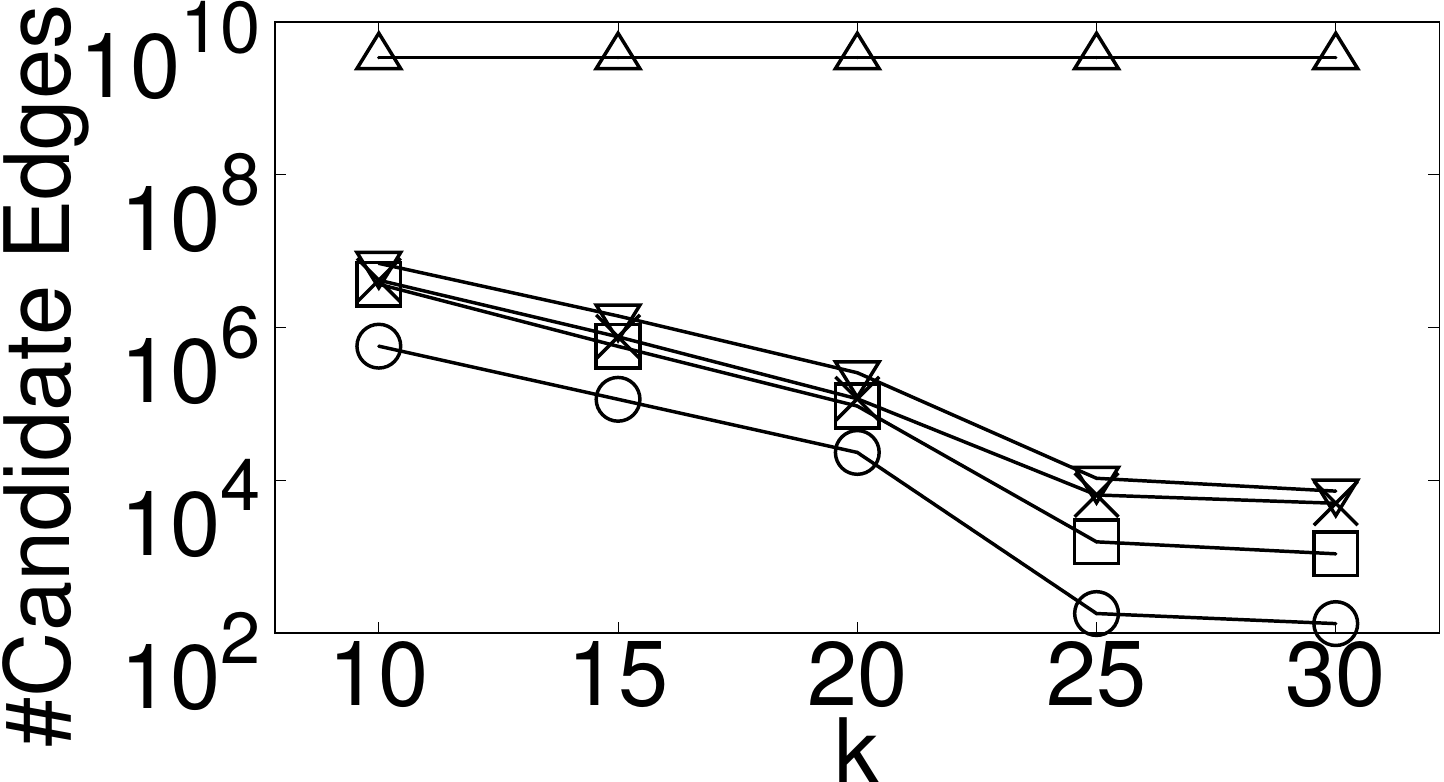}
		}
	\end{minipage}
\vspace{-3mm}
	\caption{Candidate Pruning}
	\label{fig:reducecandidate}
\end{figure}


In Figure~\ref{fig:reducecandidate}, we report the number of \ces of different algorithms by incrementally adding techniques, where the details of each algorithms is given in Table~\ref{tb:datasets}.
We can see the improvement brought by each proposed technique.

\begin{figure}[b]
	\centering
	\includegraphics[width=0.75\columnwidth]{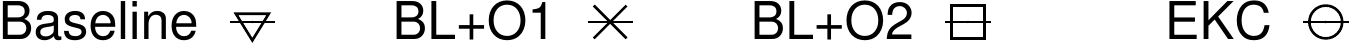}\\
	\begin{minipage}{0.47\columnwidth}
		\subfigure[Brightkite, b=5] {
			\includegraphics[width=\columnwidth,height=20mm]{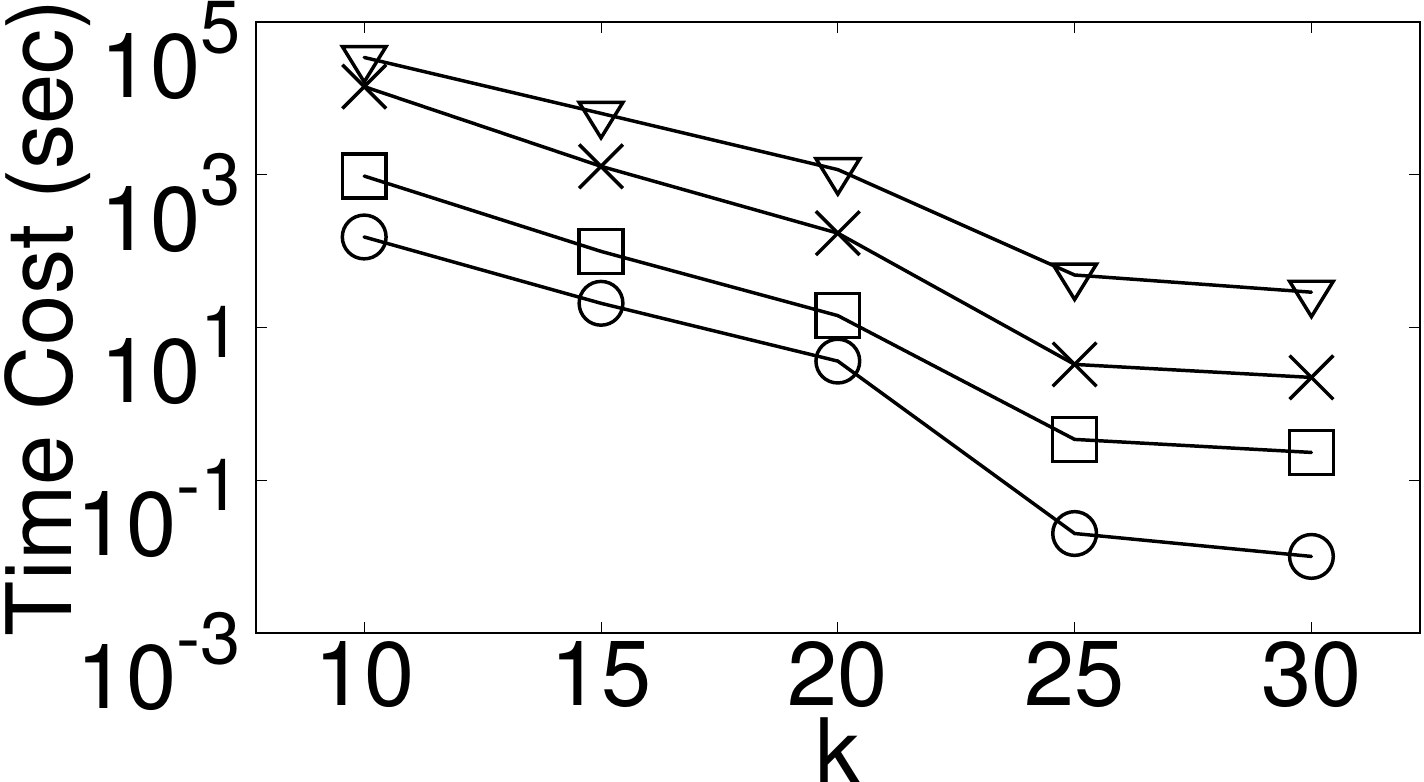}
		}
	\end{minipage}
	\begin{minipage}{0.47\columnwidth}
		\subfigure[Brightkite, k=20] {
			\includegraphics[width=\columnwidth,height=20mm]{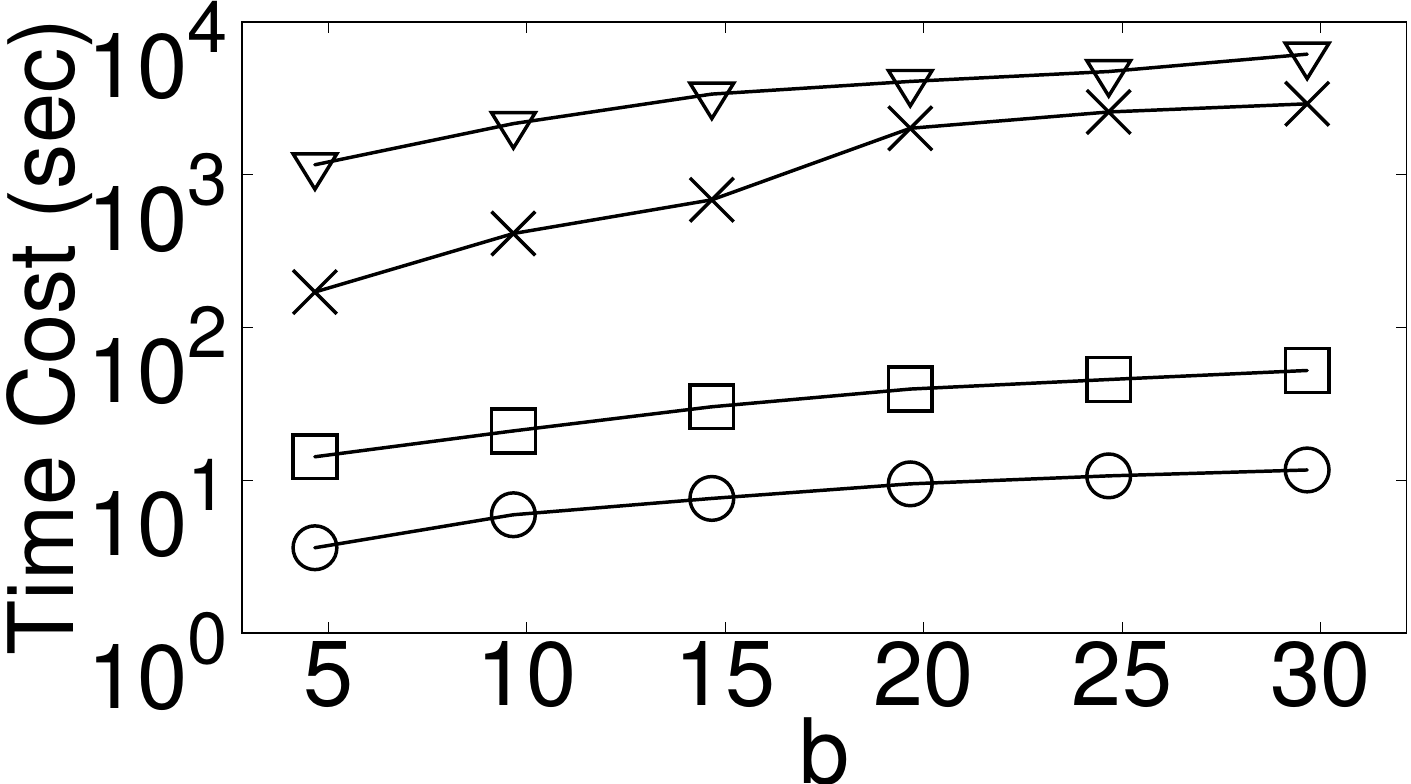}
		}
	\end{minipage}
\vspace{-3mm}
	\caption{Runtime with Different $k$ and $b$}
	\label{fig:k&r}
\end{figure}
%
Figure~\ref{fig:k&r} studies the impact of $k$ and $r$.
Figure~\ref{fig:k&r}(a) shows that the runtime decreases with a larger input of $k$. This is because the number of candidate edges become fewer with a larger $k$.
Figure~\ref{fig:k&r}(b) reports the runtime with an increasing $b$, which is proportional to the value of $b$.
The margin between \texttt{BL+OF} and \texttt{EKC} becomes smaller with a larger $k$, because the number of connected components becomes less when $k$ increases.
It is reported that \texttt{EKC} largely outperforms the other algorithms under all the settings.

\section{Conclusion}
\label{sec:conc}

In this paper, we investigate the problem of \ekc, which aims to add a set $b$ of edges in a network such that the size of the resulting \kc is maximized.
We prove the problem is NP-hard and APX-hard.
An efficient algorithm, named \texttt{EKC}, is proposed with novel optimizations. 
Extensive experiments on 9 real-life datasets are conducted to demonstrate our model is effective and our algorithm is efficient.

\section*{Acknowledgments}
Xuemin Lin is supported by 2019DH0ZX01, 2018YFB10035\\04, NSFC61232006, ARC DP180103096 and DP170101628.
Wenjie Zhang is supported by ARC DP180103096.
\bibliographystyle{named}
\bibliography{ref}

\begin{thebibliography}{}

\bibitem[\protect\citeauthoryear{Batagelj and
  Zaversnik}{2003}]{DBLP:journals/corr/cs-DS-0310049}
Vladimir Batagelj and Matjaz Zaversnik.
\newblock An o(m) algorithm for cores decomposition of networks.
\newblock {\em CoRR}, cs.DS/0310049, 2003.

\bibitem[\protect\citeauthoryear{Bhawalkar \bgroup \em et al.\egroup
  }{2015}]{DBLP:journals/siamdm/BhawalkarKLRS15}
Kshipra Bhawalkar, Jon~M. Kleinberg, Kevin Lewi, Tim Roughgarden, and Aneesh
  Sharma.
\newblock Preventing unraveling in social networks: The anchored k-core
  problem.
\newblock {\em {SIAM} J. Discrete Math.}, 29(3):1452--1475, 2015.

\bibitem[\protect\citeauthoryear{Cheema \bgroup \em et al.\egroup
  }{2010}]{DBLP:conf/icde/CheemaBLZW10}
Muhammad~Aamir Cheema, Ljiljana Brankovic, Xuemin Lin, Wenjie Zhang, and Wei
  Wang.
\newblock Multi-guarded safe zone: An effective technique to monitor moving
  circular range queries.
\newblock In {\em {ICDE}}, pages 189--200, 2010.

\bibitem[\protect\citeauthoryear{Chitnis and
  Talmon}{2018}]{DBLP:conf/csr/ChitnisT18}
Rajesh Chitnis and Nimrod Talmon.
\newblock Can we create large k-cores by adding few edges?
\newblock In {\em CSR}, pages 78--89, 2018.

\bibitem[\protect\citeauthoryear{Chitnis \bgroup \em et al.\egroup
  }{2013}]{DBLP:conf/aaai/ChitnisFG13}
Rajesh~Hemant Chitnis, Fedor~V. Fomin, and Petr~A. Golovach.
\newblock Preventing unraveling in social networks gets harder.
\newblock In {\em AAAI}, 2013.

\bibitem[\protect\citeauthoryear{Feige}{1998}]{DBLP:journals/jacm/Feige98}
Uriel Feige.
\newblock A threshold of ln \emph{n} for approximating set cover.
\newblock {\em J. {ACM}}, 45(4):634--652, 1998.

\bibitem[\protect\citeauthoryear{Kapron \bgroup \em et al.\egroup
  }{2011}]{DBLP:conf/asunam/KapronSV11}
Bruce~M. Kapron, Gautam Srivastava, and S.~Venkatesh.
\newblock Social network anonymization via edge addition.
\newblock In {\em {ASONAM}}, pages 155--162, 2011.

\bibitem[\protect\citeauthoryear{Karp}{1972}]{karp1972reducibility}
Richard~M. Karp.
\newblock Reducibility among combinatorial problems.
\newblock In {\em Complexity of Computer Computations}, pages 85--103, 1972.

\bibitem[\protect\citeauthoryear{Kitsak \bgroup \em et al.\egroup
  }{2010}]{kitsak2010identification}
Maksim Kitsak, Lazaros~K Gallos, Shlomo Havlin, Fredrik Liljeros, Lev Muchnik,
  H~Eugene Stanley, and Hern{\'a}n~A Makse.
\newblock Identification of influential spreaders in complex networks.
\newblock {\em Nature physics}, 6(11):888--893, 2010.

\bibitem[\protect\citeauthoryear{Lai \bgroup \em et al.\egroup
  }{2005}]{DBLP:journals/endm/LaiTK05}
Yung{-}Ling Lai, Chang{-}Sin Tian, and Ting{-}Chun Ko.
\newblock Edge addition number of cartesian product of paths and cycles.
\newblock {\em Electronic Notes in Discrete Mathematics}, 22:439--444, 2005.

\bibitem[\protect\citeauthoryear{Luce and Perry}{1949}]{luce1949method}
R~Duncan Luce and Albert~D Perry.
\newblock A method of matrix analysis of group structure.
\newblock {\em Psychometrika}, 14(2):95--116, 1949.

\bibitem[\protect\citeauthoryear{Luo \bgroup \em et al.\egroup
  }{2008}]{DBLP:conf/icde/LuoWL08}
Yi~Luo, Wei Wang, and Xuemin Lin.
\newblock {SPARK:} {A} keyword search engine on relational databases.
\newblock In {\em {ICDE}}, pages 1552--1555, 2008.

\bibitem[\protect\citeauthoryear{Luo \bgroup \em et al.\egroup
  }{2011}]{DBLP:journals/tkde/LuoWLZWL11}
Yi~Luo, Wei Wang, Xuemin Lin, Xiaofang Zhou, Jianmin Wang, and Keqiu Li.
\newblock {SPARK2:} top-k keyword query in relational databases.
\newblock {\em {IEEE} Trans. Knowl. Data Eng.}, 23(12):1763--1780, 2011.

\bibitem[\protect\citeauthoryear{Malliaros and
  Vazirgiannis}{2013}]{DBLP:conf/cikm/MalliarosV13}
Fragkiskos~D. Malliaros and Michalis Vazirgiannis.
\newblock To stay or not to stay: modeling engagement dynamics in social
  graphs.
\newblock In {\em {CIKM}}, pages 469--478, 2013.

\bibitem[\protect\citeauthoryear{Medya \bgroup \em et al.\egroup
  }{2019}]{DBLP:journals/corr/abs-1901-02166}
Sourav Medya, Tiyani Ma, Arlei Silva, and Ambuj~K. Singh.
\newblock K-core minimization: {A} game theoretic approach.
\newblock {\em CoRR}, abs/1901.02166, 2019.

\bibitem[\protect\citeauthoryear{Morone \bgroup \em et al.\egroup
  }{2019}]{morone2019k}
Flaviano Morone, Gino Del~Ferraro, and Hern{\'a}n~A Makse.
\newblock The k-core as a predictor of structural collapse in mutualistic
  ecosystems.
\newblock {\em Nature Physics}, 15(1):95, 2019.

\bibitem[\protect\citeauthoryear{Natanzon \bgroup \em et al.\egroup
  }{2001}]{DBLP:journals/dam/NatanzonSS01}
Assaf Natanzon, Ron Shamir, and Roded Sharan.
\newblock Complexity classification of some edge modification problems.
\newblock {\em Discrete Applied Mathematics}, 113(1):109--128, 2001.

\bibitem[\protect\citeauthoryear{Seidman}{1983}]{seidman1983network}
Stephen~B Seidman.
\newblock Network structure and minimum degree.
\newblock {\em Social networks}, 5(3):269--287, 1983.

\bibitem[\protect\citeauthoryear{Suady and Najim}{2014}]{suady2014edge}
Suaad~AA Suady and Alaa~A Najim.
\newblock On edge-addition problem.
\newblock {\em Journal of College of Education for Pure Science}, 4(1):26--36,
  2014.

\bibitem[\protect\citeauthoryear{Ugander \bgroup \em et al.\egroup
  }{2012}]{ugander2012structural}
Johan Ugander, Lars Backstrom, Cameron Marlow, and Jon Kleinberg.
\newblock Structural diversity in social contagion.
\newblock {\em PNAS}, 109(16):5962--5966, 2012.

\bibitem[\protect\citeauthoryear{Wu \bgroup \em et al.\egroup
  }{2013}]{DBLP:conf/wsdm/WuSFLT13}
Shaomei Wu, Atish~Das Sarma, Alex Fabrikant, Silvio Lattanzi, and Andrew
  Tomkins.
\newblock Arrival and departure dynamics in social networks.
\newblock In {\em WSDM}, pages 233--242, 2013.

\bibitem[\protect\citeauthoryear{Zhang \bgroup \em et al.\egroup
  }{2017a}]{DBLP:journals/pvldb/ZhangZZQL17}
Fan Zhang, Wenjie Zhang, Ying Zhang, Lu~Qin, and Xuemin Lin.
\newblock {OLAK:} an efficient algorithm to prevent unraveling in social
  networks.
\newblock {\em {PVLDB}}, 10(6):649--660, 2017.

\bibitem[\protect\citeauthoryear{Zhang \bgroup \em et al.\egroup
  }{2017b}]{DBLP:conf/aaai/ZhangZQZL17}
Fan Zhang, Ying Zhang, Lu~Qin, Wenjie Zhang, and Xuemin Lin.
\newblock Finding critical users for social network engagement: The collapsed
  k-core problem.
\newblock In {\em AAAI}, pages 245--251, 2017.

\bibitem[\protect\citeauthoryear{Zhang \bgroup \em et al.\egroup
  }{2017c}]{DBLP:conf/icde/ZhangYZQ17}
Yikai Zhang, Jeffrey~Xu Yu, Ying Zhang, and Lu~Qin.
\newblock A fast order-based approach for core maintenance.
\newblock In {\em ICDE}, pages 337--348, 2017.

\bibitem[\protect\citeauthoryear{Zhang \bgroup \em et al.\egroup
  }{2018a}]{DBLP:conf/dasfaa/ZhangYZQLZ18}
Fan Zhang, Long Yuan, Ying Zhang, Lu~Qin, Xuemin Lin, and Alexander Zhou.
\newblock Discovering strong communities with user engagement and tie strength.
\newblock In {\em {DASFAA}}, pages 425--441, 2018.

\bibitem[\protect\citeauthoryear{Zhang \bgroup \em et al.\egroup
  }{2018b}]{DBLP:conf/icde/ZhangZQZL18}
Fan Zhang, Ying Zhang, Lu~Qin, Wenjie Zhang, and Xuemin Lin.
\newblock Efficiently reinforcing social networks over user engagement and tie
  strength.
\newblock In {\em {ICDE}}, pages 557--568, 2018.

\bibitem[\protect\citeauthoryear{Zhu \bgroup \em et al.\egroup
  }{2018}]{DBLP:conf/cikm/Zhu0WL18}
Weijie Zhu, Chen Chen, Xiaoyang Wang, and Xuemin Lin.
\newblock K-core minimization: An edge manipulation approach.
\newblock In {\em {CIKM}}, pages 1667--1670, 2018.

\end{thebibliography}

\end{document}